\documentclass[a4paper,10pt,english]{article}
\usepackage[utf8]{inputenc}
\usepackage[T1]{fontenc}
\usepackage{hyperref}
\usepackage{graphicx}
\usepackage[ruled,vlined,linesnumbered]{algorithm2e}
\usepackage{bm}
\usepackage{comment}

\usepackage{amsmath}
\usepackage{mathtools}
\usepackage{amssymb}
\usepackage{tikz}
\usepackage{amsthm}
\usepackage{listings}
\usepackage{textgreek}

\newtheorem{definition}{Definition}
\newtheorem{lemma}{Lemma}
\newtheorem{theorem}{Theorem}
\newtheorem{observation}{Observation}

\newcommand{\believes}{\mid\equiv}
\newcommand{\sees}{\triangleleft}
\newcommand{\oncesaid}{\mid\sim}
\newcommand{\controls}{\Rightarrow}
\newcommand{\nonce}[1]{\#(#1)}

\newcommand{\sharekey}[1]{\xleftrightarrow{#1}}
\newcommand{\pubkey}[1]{\xmapsto{#1}}

\usepackage{stackengine,amssymb,graphicx}

\newcommand\encircle[1]{%
  \tikz[baseline=(X.base)] 
    \node (X) [draw, shape=circle, inner sep=0] {\strut #1};}

\usepackage{hyperref}
\hypersetup{
    colorlinks=true,
    linkcolor=black,
    anchorcolor=black,
    citecolor=black,
    filecolor=blue,
    menucolor=black,
    urlcolor=blue,
    bookmarks=true,
    bookmarksnumbered=true
}
\lstdefinelanguage{Lean}{
morekeywords={lemma, assume, have, from, show, inductive, Type, notation, def, theorem, axiom, Prop, example, \$},
morecomment=[l]{--}, }
\title{DELP: Dynamic Epistemic Logic for\\ Security Protocols}

\author{
    Ioana Leuștean and  Bogdan Macovei\\
    \textit{Faculty of Mathematics and Computer Science}\\
    \textit{University of Bucharest}\\
    Bucharest, Romania\\
    ioana.leustean@unibuc.ro, bogdan.macovei@unibuc.ro
}

\date{\empty} 

\begin{document}
\maketitle
\begin{abstract} 
The formal analysis of security protocols is a challenging field, with various approaches being studied nowadays. The famous Burrows-Abadi-Needham Logic was the first logical system aiming to validate security protocols. Combining ideas from previous approaches,  in this paper we define a complete system of \textit{dynamic epistemic logic} for modeling security protocols.  Our logic  is implemented, and few of its properties are verifyied, using the theorem prover Lean.
\end{abstract}

\section{Introduction}
This paper presents \textit{DELP}, a dynamic epistemic logic  for analysing security protocols.  In order to define our logic, we  combine the epistemic approach to authentification from \cite{halpern2017epistemic}, the expectation semantics from \cite{van2014hidden} and the operational semantics for security protocols from  \cite{hollestelle2005systematic}.  

Our main contributions are: (i) the definition of \textit{DELP} as a sound and complete system with respect to an expectation semantics representing the adversary knowledge; (ii) the implementation of \textit{DELP} in the theorem prover \textit{Lean}. Consequently, using \textit{Lean}: (iii) we defined translations in \textit{DELP} for a few inference rules of the  Burrows-Abadi-Needham (\textit{BAN}) logic \cite{burrows1989logic} and we proved their soundness, (iv) we defined the \textit{Needham-Schroeder} authentication protocol as a theory in \textit{DELP} and we verified a few security claims.


Section \ref{prel} presents the \textit{Needham-Schroder} security protocol and recalls the  formal approaches from \cite{halpern2017epistemic}, \cite{van2014hidden} and \cite{hollestelle2005systematic}.  In Section \ref{delp} we define the system \textit{DELP} and we prove its properties.  Section \ref{leandelp} contains the  Lean implementation  of \textit{DELP}. Few deduction rules of the \textit{BAN} Logic are defined in \textit{DELP} and their soundness is proved using the \textit{Lean} implementation.  In  Section \ref{nsprotocol} we study the \textit{Needham-Schroeder} authentication protocol using \textit{DELP} and its \textit{Lean} implementation. The last section contains conclusions and further developments. 

\section{Preliminaries: formal analysis of security protocols}\label{prel}

A \textit{security protocol} is defined as a set of rules and conventions that determine the exchange of messages between two or more agents in order to implement a security service. The protocol must be unambiguous and must allow the description of several roles, so that an agent can perform a certain role at a certain protocol round. An example of a security protocol, which we will mention and use in this paper, is the \textit{Needham-Schroeder} protocol.

\subsection{The Needham-Schroeder symmetric key protocol for key exchange}

The protocol specification for three agents is as follows:

\begin{align*}
    A \to S: &\ A, B, N_a \\
    S \to A: &\ \{ N_a, B, K_{ab}, \{ K_{ab}, A \}_{K_{bs}} \}_{K_{as}} \\
    A \to B: &\ \{ K_{ab}, A \}_{K_{bs}} \\ 
    B \to A: &\ \{ N_b \}_{K_{ab}} \\ 
    A \to B: &\ \{ N_b - 1 \}_{K_{ab}}
\end{align*}

A step-by-step description of the protocol is:
\begin{enumerate}
    \item Alice initiates the connection with the Server, sending who she is, with whom she wants to communicate and a \textit{nonce};
    \item the Server sends - encrypted with the common key between Alice and Server - the nonce generated by Alice, the identity of Bob and the communication key between Alice and Bob, to which is added a message that only Bob can decrypt (being encrypted with the communication key between Bob and Sserver), which contains the communication key shared by Alice and Bob; in this way, Alice cannot read the message sent by Server to Bob;
    \item Alice sends Bob the message that it could not decrypt, received from the Server;
    \item Bob decrypts the message, and sends Alice a \textit{nonce} encrypted with the common key between Alice and Bob;
    \item Alice receives Bob's message, decypts it, and resends it, applying a simple function to it - in this case, it decrements it. This step is useful in two situations: it is a first protection on a \textit{reply attack} and it shows that the agents are still \textit{alive} in the session.
\end{enumerate}

\subsection{BAN Logic}

We will briefly present the BAN logic, based on \cite{burrows1989logic}. The mathematical system contains the following sets: a set of participating agents in communication protocol sessions - named, generally, using capital letters of the beginning of the alphabet (A, B, ...), a set of keys - named, generally, $K_{a, b}$ for the public key between agents A and B, $K_a$ for A's public key and $K_a^{-1}$ for A's secret key, and a set of messages - named, generally, using capital letters of the end of the alphabet (X, Y, ...). An encrypted message is denoted by writing $\{X\}_k$, meaning that the message $X$ is encrypted with the key $k$.

The specific formulas introduced in BAN logic are the following: 
\begin{itemize}
    \item $P \believes X$: the agent $P$ \textbf{believes} the message $X$;
    \item $P \sees X$: the agent $P$ \textbf{sees} or \textbf{receives} $X$;
    \item $P \oncesaid X$: the agent $P$ \textbf{once said} or \textbf{sends} $X$;
    \item $P \controls X$: the agent $P$ \textbf{controls}  $X$ or \textbf{have jurisdiction} over $X$;
    \item $\nonce X$: $X$ is a \textit{nonce};
    \item $P \sharekey{k} Q$: the agents $P$ and $Q$ shares the communication key $k$;
    \item $\pubkey{k} P$: $k$ is $P$'s public key;
    \item $\{X\}_k$: $X$ is encrypted with the key $k$;
    \item $<X>_Y$: $X$ is encrypted with the common secret $Y$.
\end{itemize}

In the sequel we recall only two deductions rules, we refer to \cite{burrows1989logic} for the full deduction system.

\ 

The  \textbf{Message Meaning Rule}, formally defined by
\begin{align}
    \frac{P \believes Q \sharekey{K} P \ \ P \sees \{ X \}_K}{P \believes Q \oncesaid X}
\end{align}
can be read as follows: if agent \texttt{P} belives that he has a communication key $K$ with  agent \texttt{Q}, and agent \texttt{P} receives a message \texttt{X} encrypted under $K$, then  $P$  belives that the encrypted message was sent by \texttt{Q}.

\

The  \textbf{Jurisdiction} rule, formally defined by
\begin{align}
     \frac{P \believes Q \controls X \ \ P \believes Q \believes X}{P \believes X}
\end{align}
can be  read as follows: if agent \texttt{P} belives that agent \texttt{Q} has jurisdiction over a message \texttt{X} and, furthermore, agent \texttt{P} belives  that \texttt{Q} belives \texttt{X}, then \texttt{P} belives  \texttt{X}.

\subsection{An approach based on epistemic logic}

In this subsection, we recall the main ideas from \cite{halpern2017epistemic}, and we refer to  \cite{van2007dynamic} for a comprehensive presentation of dynamic epistemic logic. 

In this paper, there are defined $K$ (the set of communication keys), $N$ (the set of \textit{nonce}s), $T$ (the set of plain texts) and $\Phi$ (the set of formulas). The BNF specification of the language is:
\begin{align*}
    \textbf{s} &::= s \ | \ x \\
    \textbf{m} &::= t \ | \ k \ | \ n \ | \ i \ | \ (m_1, m_2) \ | \ \{m\}_k \ | \ \varphi \\
    \varphi &::= p \ | \ sent_i(s) \ | \ recv_i(s) \ | \ extract_i(m) \ | \ \neg \varphi \  | \ \varphi_1 \land \varphi_2 \ | \ K_i \varphi \ | \\ & \ \ \ \ \ \bigcirc \varphi \ | \ \encircle{-} \varphi \ | \ \Box \varphi \ |  \ \boxed{-} \varphi \ | \ \exists x \varphi \ | \ [m] = s \ | \ s \sqsubseteq s' \ | \ Pr_i(\varphi) \geq \alpha
\end{align*}
where $p$ is an atomic formula, $i$ is an arbitrary agent, $m$ is an arbitrary message, $t \in T$, $k \in K$, $n \in N$, $\alpha \in [0, 1]$ a probability, $s$ a string, $x$ a variable over strings and $\varphi \in \Phi$. 

For semantics, the models are 
$$I = (R, \pi, \mathbf{C}, \{\mu_C\}_{C \in \mathbf{C}})$$
where $R$ is a protocol rounds system, $\pi$ is an evaluation function, $\mathbf{C}$ is a partition of $R$, and for every $C \in \mathbf{C}$, the measure $\mu_C$ is the distribution probability over rounds in $C$. The inductive interpretation of formulas in this models are:
\begin{align*}
    &(I, r, m) \models p \Longleftrightarrow \pi(r(m))(p) \text{ is true} \\
    &(I, r, m) \models \neg \varphi \Longleftrightarrow (I, r, m) \not \models \varphi \\
    &(I, r, m) \models \varphi_1 \land \varphi_2 \Longleftrightarrow (I, r, m) \models \varphi_1 \text{ and } (I, r, m) \models \varphi_2 \\
    &(I, r, m) \models K_i \varphi \Longleftrightarrow  \text{ for all } (r', m') \sim_i (r, m), \\ & \ \ \ \ \ \ \ \ \ \ \ \ \ \ \text{we have } (I, r', m') \models \varphi \\
    &(I, r, m) \models \bigcirc \varphi \Longleftrightarrow (I, r, m+1) \models \varphi \\
    &(I, r, m) \models \encircle{-} \varphi \Longleftrightarrow m = 0 \text{ or } (I, r, m-1) \models \varphi \\
    &(I, r, m) \models \Box \varphi \Longleftrightarrow \text{ for all } m' \geq m, (I, r, m') \models \varphi \\
    &(I, r, m) \models \boxed{-} \varphi \Longleftrightarrow \text { for all } \ m' \leq m, (I, r, m') \models \varphi \\
    &(I, r, m) \models Pr_i(\varphi) \geq \alpha \Longleftrightarrow\\ & \ \ \ \  \mu_{r,m,i}(\{(r', m') \ | \ (I, r', m') \models \varphi \} \cap K_i(r, m) \cap \mathbf(C)(r)) \geq \alpha \\
    &(I, r, m) \models \exists x \varphi \Longleftrightarrow \text{ exists } s \text{ string}, (I, r, m) \models \varphi[s/x]
\end{align*}

\subsection{An approach based on operational semantics}
From \cite{hollestelle2005systematic}, the main point of interest is the terms deduction system. In this formal system we have terms (roles, messages, keys and \textit{nonce}s), variables over \textit{Var}, \textit{Fresh} and \textit{Role} sorts, functions symbols (in \textit{Func}), the protocols specifications and a labeled transition system for the execution of the protocols.

Having $\Gamma$ a knowledge set, the term deduction rules are:
\begin{itemize}
    \item if $t \in \Gamma$, then $\Gamma \vdash t$;
    \item $\Gamma \vdash t_1$ and $\Gamma \vdash t_2$ if and only if $\Gamma \vdash (t_1, t_2)$;
    \item if $\Gamma \vdash t$ and $\Gamma \vdash k$, then $\Gamma \vdash \{t\}_k$;
    \item if $\Gamma \vdash \{t\}_k$ and $\Gamma \vdash k^{-1}$, then $\Gamma \vdash t$;
    \item if $\Gamma \vdash t_i$, $1 \leq 1 \leq n$, then $\Gamma \vdash f(t_1, t_2, ..., t_n)$, where $f$ is a function symbol of \textit{Func}, with the arity $n$.
\end{itemize}

\subsection{An approach based on expectation models}

In this subsection, we will present the main results of \cite{van2014hidden}, that we will use in the next section to prove the completeness theorem of our system.

In this paper there are introduced two sets, $I$ - the set of agents and $P$ - the set of formulas. For interpreting formulas there are used \textit{Kripke} models, $\mathcal{M} = (S, \sim, V)$, where $S$ is the set of accessible world, $\sim$ is the accessibility relation between worlds and $V$ is the evaluation function, $V : P \to \mathcal{P}(S)$.

There are an action set - $\Sigma$ - and a langue of observations - $\mathcal{L}_{obs}$. The BNF grammar of the actions is:
\begin{align}
    \pi ::= \delta \ | \ \varepsilon \ | \ a \ | \ \pi \cdot \pi \ | \ \pi + \pi \ | \ \pi^* 
\end{align}
where $\delta$ is an empty set of observations, $\varepsilon$ is the empty string and $a \in \Sigma$.

The observations set is denoted by $\mathcal{L}(\pi)$ and is inductively defined as:
\begin{align}
    &\mathcal{L}(\delta) = \emptyset \\
    &\mathcal{L}(\varepsilon) = \{ \varepsilon \} \\ 
    &\mathcal{L}(a) = \{ a \} \\ 
    &\mathcal{L}(\pi \cdot \pi') = \{ wv | w \in \mathcal{L}(\pi) \text{ and } v \in \mathcal{L}(\pi')\} \\
    &\mathcal{L}(\pi + \pi') = \mathcal{L}(\pi) \cup \mathcal{L}(\pi') \\ 
    &\mathcal{L}(\pi^*) = \{ \varepsilon \} \cup \bigcup_{n > 0}(\mathcal{L}(\pi \cdot ... \cdot \pi))
\end{align}

An epistemic model defined with this observations is an epistemic expectation model $\mathcal{M} = (S, \sim, V, Exp)$, where $Exp : S \to \mathcal{L}_{obs}$ is a function that maps every state from $S$ to an observation $\pi$ for which $\mathcal{L}(\pi) \neq \emptyset$. The logical formulas are defined using the following BNF description:
\begin{align}
    \varphi ::= p \ | \ \neg \varphi \ | \ \varphi \land \psi \ | \ K_i \varphi \ | \ [\pi]\varphi
\end{align}
where $p \in P$, $i \in I$ and $\pi \in \mathcal{L}_{obs}$.

An important result from this paper is the \textit{bisimilarity}; a binary relation $R$ between two epistemic expectations models $\mathcal{M} = (S, \sim, V, Exp)$ and $\mathcal{N} = (S', \sim', V', Exp')$ is called bisimilarity if for every $s \in S$ and $s' \in S'$, if we have $(s, s') \in R$, then: 
\begin{align}
    &\textbf{Propositional invariance: } 
     V(s) = V'(s') \\ 
    &\textbf{Observation invariance: } 
       \mathcal{L}(Exp(s)) = \mathcal{L}(Exp(s')) \\ 
    &\textbf{Zig: } s \sim_i t \in \mathcal{M} \Longrightarrow \text{exists } t' \in \mathcal{N} \\ 
    & \nonumber \ \ \ \ \ \ \ \ \ \ \ \ \ \ \ \ \ \ \ \ \text{ such that } s' \sim_i' t' \text{ and } t R t' \\ 
    &\textbf{Zag: } s' \sim_i' t' \in \mathcal{N} \Longrightarrow \text{exists } t \in \mathcal{M} \\ 
    & \nonumber \ \ \ \ \ \ \ \ \ \ \ \ \ \ \ \ \ \ \ \ \text{ such that } s \sim_i t \text{ and } t R t'
\end{align}

The article also introduce the \textit{bisimilarity invariance}: for two epistemic states $\mathcal{M}, s$ and $\mathcal{N}, s'$, the following two statements are equivalent:
\begin{align}
     &i) \ \mathcal{M}, s \leftrightarrow \mathcal{N}, s' \\ 
    &ii) \ \text{for all } \varphi \text{: } \mathcal{M}, s \models \varphi \Longleftrightarrow \mathcal{N}, s' \models \varphi
\end{align}

\textbf{Updated models}. Let $w$ be an observation over $\Sigma$, and $\mathcal{M} = (S, \sim, V, Exp)$ an epistemic expectation model. The, the \textbf{updated} model is denoted with $\mathcal{M}|_w = (S', \sim', V', Exp')$, where $S' = \{ s \ | \ \mathcal{L}(Exp(s) - w) \neq \emptyset \}$, $\sim_i' = \sim_i|_{S' \times I \times S'}$, $V' = V|_{S'}$ and $Exp'(s) = Exp(s) - w$, where $\pi - w = \{ v \ | \ wv \in \mathcal{L}(\pi) \}$.

\textbf{Temporal models}. Let $\mathcal{M} = (S, \sim, V, Exp)$ be an epistemic expectation model. Then the temporal model is called $ET(\mathcal{M})$ and is defined as $ET(\mathcal{M}) = (H, \to_a, \sim_i', V')$, where $H = \{(s, w) \ | \ s \in S, w = \varepsilon \text{ or } w \in \mathcal{L}(Exp(s)) \}$, $(s, w) \to_a (t, v) \Longleftrightarrow s = t \text{ and } v = wa, a \in \Sigma$, $(s, w) \sim_i (t, v) \Longleftrightarrow s \sim_i t \text{ and } w = v$ and $p \in V'(s, w) \Longleftrightarrow p \in V(s)$.

Using temporal models, is it proved in this paper that  $\mathcal{M}, s \models \varphi \Longleftrightarrow ET(\mathcal{M}), (s, \varepsilon) \models_{EPDL} \varphi$, so the system is complete by the completeness of dynamic epistemic logic.

\section{DELP - Dynamic Epistemic Logic for Protocols}\label{delp}

In order to define our system, we firstly  recall the \textit{dynamic epistemic logic}  \cite{van2007dynamic}.
\textit{Dynamic epistemic logic} is a \textit{dynamic logic} \cite{harel2001dynamic} to which is added the knowledge operator \textit{K} from \textit{epistemic logic}. There are two sets, $\Pi$ - the set of programs, and $\Phi$ - the set of formulas, with $\Pi_0$ - set of atomic programs, and $\Phi_0$ - set of atomic formulas. The language is described using the following BNF:
\begin{align}
    \varphi &::= p \text{ \textbar{} } 
        \neg \varphi \text{ \textbar{} } 
        \varphi \to \varphi \text{ \textbar{} } 
        K_i \varphi \text{ \textbar{} } 
        [\alpha]\varphi
\end{align}
where $p \in \Phi_0$, $\varphi \in \Phi$, $i$ is an arbitrary agent and $\alpha \in \Pi$.

The evaluation models are \textit{Kripke models} $\mathcal{M} = (R, \sim, V)$, where $R$ is the finite set of accessible worlds, $\sim$ is the accessibility relationship between worlds, and $V$ is the evaluation from dynamic logic: for a formula $\varphi \in \Phi$, $V(\varphi) \subseteq R$, and for a program $\alpha \in \Pi$, $V(\pi) \subseteq R \times R$. 

Interpretation of formulas in this models are inductively defined as:
\begin{align}
    &\mathcal{M}, s \models p \Longleftrightarrow v \in V(s) \\ 
    &\mathcal{M}, s \models \varphi \land \psi \Longleftrightarrow \mathcal{M}, s \models \varphi \text{ and } \mathcal{M}, s \models \psi \\ 
    &\mathcal{M}, s \models \neg \varphi \Longleftrightarrow \mathcal{M}, s \not \models \varphi \\ 
    &\mathcal{M}, s \models K_i \varphi \Longleftrightarrow \text{for all } t \text{ such that } s \sim_i t, \\ \nonumber & \ \ \ \ \ \ \ \ \ \text{ we have } \mathcal{M}, t \models \varphi \\ 
    &\mathcal{M}, s \models [\alpha]\varphi \Longleftrightarrow \text{for all } t \in R \text { such that } \\ \nonumber & \ \ \ \ \ \ \ \ \ (s, t) \in V(\alpha), \text{ we have } \mathcal{M}, t \models \varphi
\end{align}

We also have the following operators for programs:
\begin{align}
 V(\alpha_1 \cup \alpha_2) &= V(\alpha_1) \cup V(\alpha_2) \\ 
 V(\alpha_1; \alpha_2) &= V(\alpha_1) \circ V(\alpha_2) \\ 
 V(\alpha^*) &= \bigcup_{n \geq 0} V(\alpha)^n
\end{align}

The deductive system contains all instances of propositional tautologies to which are added the following axioms:
\begin{align}
    &K_a (\varphi \to \psi) \to (K_a \varphi \to K_a \psi) \\ 
    &K_a \varphi \to \varphi \\ 
    &K_a \varphi \to K_a K_a \varphi \\ 
    &\neg K_a \varphi \to K_a \neg K_a \varphi \\ 
    &[\alpha](\varphi \to \psi) \to ([\alpha]\varphi \to [\alpha]\psi) \\ 
    &[\alpha](\varphi \land \psi) \leftrightarrow [\alpha]\varphi \land [\alpha]\psi \\ 
    &[\alpha \cup \beta]\varphi \leftrightarrow [\alpha]\varphi \land [\alpha]\psi \\ 
    &[\alpha; \beta]\varphi \leftrightarrow [\alpha][\beta]\varphi
\end{align}

Deductive rules are \textit{modus ponens}, \textit{generalization} from dynamic logic and \textit{necessity} from epistemic logic:
\begin{align*}
    (MP)\frac{\varphi \ \varphi \to \psi}{\psi}; \ (GEN) \frac{\varphi}{[\alpha]\varphi}; \ (NEC) \frac{\varphi}{K_i \varphi}
\end{align*}

This system is known as the \textit{PA}-system in \cite{van2007dynamic}, and it is proved sound and complete \cite[p.~187-188]{van2007dynamic}.

\subsection{DELP}
In this subsection we define \textit{DELP}, a logic  based on dynamic epistemic logic, enriched with a set of actions collected during the execution of the protocol and a  grammar for messages, together with a system of deduction for knowledge based on actions. 

\subsubsection{Syntax}
Let $Agent$ be the set of agents and let $Func$ be a set of (encryption) functions.  We consider the sets $\Phi$ and $\Pi$ like in dynamic epistemic logic, with $\Phi_0$ the set of \textit{atomic formulas}, and $\Pi_0$ defined by 
\begin{align}
    \Pi_0 := \{ send_i, recv_i\}|_{i \in Agent}
\end{align}

\noindent The elements of $\Pi_0$ are \textit {protocols actions}:  we read $send_i$ as "the agent $i$ sends" and we read $recv_i$ as "the agent $i$ receives".

In the following we define \textit{messages} and \textit{formulas}. In a security protocol, a message  contains  clear texts, keys, \textit{nonce}s, and agents identities. The possible operations are messages concatenation and messages encryption. Following \cite{hollestelle2005systematic}, the grammar for messages is:

\begin{align}
    m &::= text(m) \text{ \textbar{} } 
        key_{m}(i, j) \text{ \textbar{} } 
        nonce(m) \text{ \textbar{} } 
        agent(i) \\
        & \ \ \ \ \ \text{ \textbar{} } 
        (m, m) \text{ \textbar{} } 
        \{m\}_{m} \text{ \textbar{} } 
        f (m,\ldots,m)
\end{align}
where $i, j \in Agent$ and $f\in Func$. In the sequel we will use 
$t$ for texts, $k$ for keys, $n$ for nonces and $i$, $j$ for agents.  Based on  \cite{hollestelle2005systematic}, we define the following deductive system on messages:
\begin{align}
    &\dfrac{}{nonce(m)} \ \ \
    \dfrac{key_k(i, j)}{key_k(j, i)} \ \ \ 
    \dfrac{m_1 \ \ \ m_2}{(m_1, m_2)} \ \ \ 
    \\ \nonumber
    \\ \nonumber
    &\dfrac{t \ \ k}{\{t\}_k} \ \ \ 
    \dfrac{\{t\}_k \ \ k}{t} \ \ \ 
    \dfrac{t_1, t_2, ..., t_n}{f(t_1, t_2, ..., t_n)}
\end{align}
\medskip 

Finally, we are able to define the \textit{DELP} formulas:
\begin{align}
    \varphi &::= p \text{ \textbar{} } 
        \neg \varphi \text{ \textbar{} } 
        \varphi \to \varphi \text{ \textbar{} } 
        K_i \varphi \text{ \textbar{} } 
        [\alpha]\varphi \text{ \textbar{} } 
        @\mu
\end{align}

\noindent  Note that our formulas are the usual formulas of dynamic epistemic logic with protocol actions instead of programs, endowed with the  $@$-operator which converts a  message into a formula. 

\subsubsection{Semantics}

The models that we use are \textit{Kripke} models like in dynamic epistemic logic, $\mathcal{M} = (R, \sim, V)$ which we extend with $Exp$ set, a knowledge set with information collected from protocol runs. 
\begin{definition}
Let $\mathcal{M} = (R, \sim, V, Exp)$ be a DELP model, where
\begin{enumerate}
    \item $R$ is the finite set of accessible worlds;
    \item $\sim := \bigcup_{i \in Agent} \sim_i$ represents the accessibility relationship between worlds, based on epistemic relation;
    \item $V$ is the evaluation function from dynamic logic: $V(\varphi) \subseteq R$ for any $\varphi \in \Phi$, and $V(\alpha) \subseteq R \times R$, for any $\alpha \in \Pi$;
    \item $Exp$ is the knowledge set: for any $s \in R$, $Exp(s)$ represents the set of all knowledge inferred up to \textit{s}-th round of the protocol;
    \item for any agent $i$, $V(send_i) \subseteq \sim_i$ and $V(recv_i) \subseteq \sim_i$.
\end{enumerate}
\end{definition}

Having this models, we can interpret $@\mu$ formula as:
\begin{align}
    \mathcal{M}, s \models @\mu \Longleftrightarrow \mu \in Exp(s)
\end{align}
The other formulas have the interpretation from the dynamic epistemic logic:
\begin{align}
    &\mathcal{M}, s \models p \Longleftrightarrow v \in V(s) \\ 
    &\mathcal{M}, s \models \varphi \land \psi \Longleftrightarrow \mathcal{M}, s \models \varphi \text{ and } \mathcal{M}, s \models \psi \\ 
    &\mathcal{M}, s \models \neg \varphi \Longleftrightarrow \mathcal{M}, s \not \models \varphi \\ 
    &\mathcal{M}, s \models K_i \varphi \Longleftrightarrow \text{for all } t \text{ such that } s \sim_i t, \\ \nonumber & \ \ \ \ \ \ \ \ \  \text{ we have } \mathcal{M}, t \models \varphi \\ 
    &\mathcal{M}, s \models [\alpha]\varphi \Longleftrightarrow \text{for all } t \in R \text { such that } \\ \nonumber & \ \ \ \ \ \ \ \ \ (s, t) \in V(\alpha), \text{ we have } \mathcal{M}, t \models \varphi
\end{align}

\subsubsection{Deductive system}
The deductive system contains all instances of propositional tautologies to which are added the following axioms from dynamic epistemic logic:
\begin{align}
    &K_a (\varphi \to \psi) \to (K_a \varphi \to K_a \psi) \\ 
    &K_a \varphi \to \varphi \\ 
    &K_a \varphi \to K_a K_a \varphi \\ 
    &\neg K_a \varphi \to K_a \neg K_a \varphi \\ 
    &[\alpha](\varphi \to \psi) \to ([\alpha]\varphi \to [\alpha]\psi) \\ 
    &[\alpha](\varphi \land \psi) \leftrightarrow [\alpha]\varphi \land [\alpha]\psi \\ 
    &[\alpha \cup \beta]\varphi \leftrightarrow [\alpha]\varphi \land [\alpha]\psi \\ 
    &[\alpha; \beta]\varphi \leftrightarrow [\alpha][\beta]\varphi
\end{align}
In addition, we have the following specific axiom, that is necessary to have a correspondence between states; if the agent \textit{i} performs an action within the protocols (sends or receives a message), then he knows the message:
\begin{align}
[send_i]@m \lor [recv_i]@m \to K_i @m
\end{align}


The soundness of this system is given by the soundness of the dynamic epistemic logic \cite[p.~187-188]{van2007dynamic}, and all that remains for us to prove is the soundness of the specific axiom.


\begin{lemma}
Axiom $[send_i]@m \lor [recv_i]@m \to K_i @m$ is sound.
\end{lemma}
\begin{proof}
Let $\mathcal{M} = (R, \sim, V, Exp)$ be a \textit{DELP} model and $s \in R$ an arbitrary state. 
\begin{align*}
    \mathcal{M}, s &\models [send_i]@m \Longleftrightarrow \text{for all } t \text { such that } (s, t) \in V(send_i), \\ &\text{we have that } \mathcal{M}, t \models @m
\end{align*}
but $V(send_i) \subseteq \sim_i$, so 
\begin{align*}
    \mathcal{M}, s &\models [send_i]@m \Longleftrightarrow \text{for all } t \text { such that } (s, t) \in \sim_i, \\ &\text{we have that } \mathcal{M}, t \models @m \\ 
    &\Longleftrightarrow \mathcal{M}, s \models K_i@m
\end{align*}
\end{proof}


\subsubsection{Completeness}
In order to prove the completeness of \textit{DELP}, we follow ideas from  \cite{van2014hidden}  and general results from dynamic epistemic logic. 

\begin{definition}
\textbf{[Restricted model]} Let $\mu$ be a message and $\mathcal{M} = (R, \sim, V, Exp)$ a DELP model. Then, the \textbf{restricted model} is defined as
$$M|_{\mu} = (R', \sim', V', Exp')$$
where $R' = \{s \ | \ Exp(s)-\mu \neq \emptyset \}$, $\sim_i' = \sim_i|_{R' \times R'}$, $V' = V|_{R'}$, and $Exp'(s) = Exp(s) - \mu$.
\end{definition}

\begin{definition}
\textbf{[Temporal model]} Let $\mathcal{M} = (R, \sim, V, Exp)$ be a DELP model. We define
$$ET(\mathcal{M}) = (H, \to, \sim', V')$$
where
\begin{itemize}
    \item $H = \{(s, m) \ | \ s \in R, \ m \in Exp(s)\}$;
    \item $(s, m) \to (s', m')$ if and only if $s = s'$ and 
    $\{m\} \vdash m'$ using the deduction system (37);
    \item $(s, m) \sim' (s', m')$ if and only if $s \sim s'$ and $m \equiv m'$ where $\equiv$ is the logic equivalence;
    \item $p \in V'(s, m)$ if and only if $p \in V(s)$
\end{itemize}
\end{definition}

Having $\mathcal{N} = ET(\mathcal{M})$ a temporal model, we inductively define the following interpretation of formulas:
\begin{align}
    &\mathcal{N}, w \models p \Longleftrightarrow p \in V(w) \\
    &\mathcal{N}, w \models \neg \varphi \Longleftrightarrow \mathcal{N}, w \not \models \varphi \\
    &\mathcal{N}, w \models \varphi \land \psi \Longleftrightarrow \mathcal{N}, w \models \varphi \text{ and } \mathcal{N}, w \models \psi \\ 
    &\mathcal{N}, w \models K_i \varphi \Longleftrightarrow \text{for all } v \in \mathcal{N}, \text{ if } w \sim_i v, \\ \nonumber & \ \ \ \ \ \ \ \ \ \text{ then } \mathcal{N}, v \models \varphi \\ 
    &\mathcal{N}, w \models [\alpha]\varphi \Longleftrightarrow \text{for all } \mu \in Exp(\alpha), \\ \nonumber & \ \ \ \ \ \ \ \ \  w \to v \text{ implies } \mathcal{N}, w \models \varphi
\end{align}

\begin{definition}
\textbf{[Bisimilarity]} Based on \cite[~Def. 11]{van2014hidden}, we have that the binary relation $\rho \subseteq \mathcal{M} \times \mathcal{N}$, for two DELP models $\mathcal{M} = (R, \sim, V, Exp)$ and $\mathcal{N} = (R', \sim', V', Exp')$ is called bisimilarity if for any $v \in R$ and $v' \in R'$, if we have $v \rho v'$, then:
\begin{align}
    &\textbf{Propositional invariance } \\
    & \nonumber \ \ \ \ \ \ \ \ \ \ \ \ \ \ \ \ \ \ \ \ V(v) = V'(v') \\ 
    &\textbf{Observation invariance } \\ 
    & \nonumber \ \ \ \ \ \ \ \ \ \ \ \ \ \ \ \ \ \ \ \ Exp(v) = Exp(v') \\ 
    &\textbf{Zig } v \sim_i w \in \mathcal{M} \Longrightarrow \text{exists } w' \in \mathcal{N} \\ 
    & \nonumber \ \ \ \ \ \ \ \ \ \ \ \ \ \ \ \ \ \ \ \ \text{ such that } v' \sim_i' w' \text{ and } w \rho w' \\ 
    &\textbf{Zag } v' \sim_i' w' \in \mathcal{N} \Longrightarrow \text{exists } w \in \mathcal{M} \\ 
    & \nonumber \ \ \ \ \ \ \ \ \ \ \ \ \ \ \ \ \ \ \ \ \text{ such that } v \sim_i w \text{ and } w \rho w'
\end{align}
\end{definition}

\begin{theorem}
\textbf{[Bisimilarity invariance]} For two DELP states $\mathcal{M}, v$ and $\mathcal{N}, v'$, the following two statements are equivalent:
\begin{align}
    &(i) \ \mathcal{M}, v \leftrightarrow \mathcal{N}, v' \\ 
    &(ii) \ \text{for all } \varphi \text{: } \mathcal{M}, v \models \varphi \Longleftrightarrow \mathcal{N}, v' \models \varphi
\end{align}
\end{theorem}
The proof is the same as \cite[~Prop. 12]{van2014hidden}.

\begin{theorem}
\textbf{[Completeness]} Let $\mathcal{M} = (R, \sim, V, Exp)$ be a DELP model, $\varepsilon$ the initial knowledge and $\varphi \in \Phi$ a formula. Then
\begin{align}
    \mathcal{M}, v \models \varphi \Longleftrightarrow ET(\mathcal{M}), (s, \varepsilon) \models \varphi
\end{align}
\end{theorem}
\begin{proof}
We follow the proof from \cite[~Prop. 14]{van2014hidden}. The booleean and epistemic cases are immediate from the temporal model construction. For $\varphi := [\alpha]\psi$ we assume that $\mathcal{M}, v \models [\alpha]\psi$, but $ET(\mathcal{M}), (v, \epsilon) \not \models [\alpha]\psi$. Then, exists $m \in Exp(v)$ such that $ET(\mathcal{M}), (v, m) \not \models \psi$. From the construction of $ET(\mathcal{M})$, the definition of worlds is $H = \{(s, m) \ | \ s \in R, \ m \in Exp(s)\}$, so $m \in Exp(v)$. But $m$ is a message, then exists the restricted model $\mathcal{M}|_m$. From bisimilarity, we have that $ET(\mathcal{M}|_m), (v, \epsilon)$ is bisimilar with $ET(\mathcal{M}), (v, m)$. Then $ET(\mathcal{M}|_m), (v, \varepsilon) \models \neg \psi$. From the induction hypothesis, we have $\mathcal{M}, v \models \neg \psi$, which contradicts $\mathcal{M}, v \models [\alpha]\psi$. 
\end{proof}
We have that the \textit{DELP} system is complete.

\section{Implementation in Lean}\label{leandelp}

In this section we will present the implementation of our system in \textit{Lean} \cite{lean} prover assistant based on \cite{bentzen2019henkin}, and then we will prove the corectness of \textit{BAN} deduction rules in \textit{DELP}.

\subsection{Language}

To implement \textit{DELP}, we have the following inductive types:

1. For messages:
\begin{lstlisting}[language=Lean, numbers=none, basicstyle=\small]
inductive message (σ : ℕ) : Type 
  | null : fin σ -> message 
  | nonc : message -> message 
  | keys : message -> message -> message -> message 
  | encr : message -> message -> message
  | decr : message -> message -> message
  | tupl : message -> message -> message 
\end{lstlisting}
2. For programs:
\begin{lstlisting}[language=Lean, numbers=none, basicstyle=\small]
inductive program (σ : ℕ) : Type 
  | skip : program 
  | secv : program -> program -> program 
  | reun : program -> program -> program 
  | send : message σ -> program 
  | recv : message σ -> program 
\end{lstlisting}
3. For formulas:
\begin{lstlisting}[language=Lean, numbers=none, basicstyle=\small]
inductive form (σ : ℕ) : Type 
  | atom : fin σ -> form 
  | botm : form 
  | impl : form -> form -> form 
  | know : message σ -> form -> form 
  | prog : program σ -> form -> form 
  | mesg : message σ -> form 
  | and : form -> form -> form 
  | or : form -> form -> form 
\end{lstlisting}

We make the following notations:
\begin{lstlisting}[language=Lean, numbers=none, basicstyle=\small]
notation p `->` q := form.impl p q 
notation `ι` μ := form.mesg μ  
notation p `∧` q := form.and p q 
notation p `∨` q := form.or p q 
notation `K` m `,` p := form.know m p
notation `[` α `]` φ := form.prog α φ

notation `⬝` := {}
notation Γ ` ∪ ` p := set.insert p Γ

notation m `||` n := message.tupl m n 
notation `{` m `}` k := message.encr m k 
\end{lstlisting}

\subsection{Deductive system}

In order to be able to check security properties using DELP, we have two add two deduction hypotheses that help us specify symmetric key protocols:
\begin{align}
    &@\{m\}_k \land @key_k(i, j) \to [send_i]@m \lor [send_j]@m \\
    &@key_k(i, j) \to K_i @k \lor K_j @k
\end{align}

\begin{observation}
The first deduction hypothesis of the system represents a rule of \textit{honesty} of the participating agents; its need is highlighted in the modeling of the BAN logic: if there is an encrypted message with the communication key $k$, and the communication key $k$ is a key known to the agents $i$ and $j$, then the message is transmitted by only one of them.
\end{observation}
\begin{observation}
The second deduction hypothesis is a rule for modeling symmetric key protocols: if the $k$ key is a communication key between $i$ and $j$, then each of them knows it.
\end{observation}

We define the following context, a set $\Gamma$ of statements:
\begin{lstlisting}[language=Lean, numbers=none, basicstyle=\small]
def ctx (σ : ℕ) : Type := set (form σ)
\end{lstlisting}

The deductive system is:
\begin{lstlisting}[language=Lean, numbers=none, basicstyle=\small]
inductive proof (σ : ℕ) : ctx σ -> form σ -> Prop 
  | ax { Γ } { p } (h : p ∈ Γ) : proof Γ p 
  | kand { Γ } { i : message σ } { p q : form σ } : proof Γ (((K i, p) ∧ (K i, q)) -> (K i, (p ∧ q)))
  | ktruth { Γ } { i : message σ } { φ : form σ } : proof Γ ((K i, φ) -> φ)
  | kdist { Γ } { i : message σ } { φ ψ : form σ } : proof Γ ((K i, (φ -> ψ)) -> ((K i, φ) -> (K i, ψ)))
  | progrdistr { Γ } { α : program σ } { φ ψ : form σ } : proof Γ ([α](φ -> ψ) -> ([α]ψ -> [α]ψ))
  | pdtruth { Γ } { α : program σ } { φ : form σ } : proof Γ (([α]φ) -> φ)
  | honestyright { Γ } { m k i j : message σ } : proof Γ ((ι (k.keys i j)) ∧ (ι ({ m } k)) -> ([send j](ι m)))
  | knowreceive { Γ } { m i : message σ } : proof Γ (([recv i](ι m)) -> (K i, (ι m)))
  | knowsend { Γ } { m i : message σ } : proof Γ (([send i](ι m)) -> (K i, (ι m)))
  | knowreceivef { Γ } { i : message σ } { φ : form σ } : proof Γ (([recv i]φ) -> (K i, φ))
  | knowsendf { Γ } { i : message σ } { φ : form σ } : proof Γ (([send i]φ) -> (K i, φ))
  | mp { Γ } { p q : form σ } (hpq : proof Γ (p -> q)) (hp : proof Γ p) : proof Γ q
  | kgen { Γ } { φ : form σ } { i : message σ } (h : proof Γ φ) : proof Γ (K i, φ)
  | pdgen { Γ } { φ : form σ } { α : program σ } (h : proof Γ φ) : proof Γ ([α]φ)
\end{lstlisting}

\subsection{BAN Rules Verification}

In order to be able to verify the corectness of the  BAN rules, we  translate them our logic. We use the following correspondence:
\begin{enumerate}
    \item formula $i \believes m$ is translated as $K_i @m$ and it means $i$ knows $m$ in current state;
    \item formula $i \sees m$ means that $i$ receives $m$ and is translated as $[recv_i]@m$;
    \item formula $i \oncesaid m$ is translated as $[send_i]@m$;
    \item formula $i \controls m$ means that $i$ has jurisdiction over $m$, so the agent knows $m$ and $m$ is true: $K_i @m \to @m$;
    \item formula $i \sharekey{k} j$ is translated as $@key_k(i, j)$;
    \item formula $\nonce{m}$ is translated as $@nonce(m)$.
\end{enumerate}

Now, we can prove that the translations in \textit{DELP} of the most important BAN inference rules (according to \cite{halpern2017epistemic}) are sound. In the sequel, using Lean, we give the proofs only for  the \textit{Message Meaning} rule and for the \textit{Jurisdiction} rule, few other rules are analysed in the Appendix. 

\begin{lemma}
The Message Meaning rule for shared key is a correct rule in the \textit{DELP} system.
$$ \frac{i \believes j \sharekey{k} i \ \ i \sees \{m\}_k}{i \believes j \oncesaid m}$$
 \end{lemma}
\begin{proof}
We will prove this using Lean. 
\begin{lstlisting}[language=Lean, numbers=none, basicstyle=\small]
lemma MMSK_is_correct (σ : ℕ) { m k i j : message σ } { Γ : ctx σ } 
  : (σ-Γ ⊢ ((K i, (ι (k.keys i j))) ∧ ([recv i](ι { m } k)))) 
  -> (σ-Γ ⊢ (K i, ([send j](ι m)))) :=
  λ h, kgen 
    $ mp honestyright
      $ mp ktruth 
        $ mp kand
          $ andintro 
            (andleft h) 
            (mp knowreceive $ andright h). 
\end{lstlisting}
\end{proof}

A much easier demonstration is for the \textit{jurisdiction} rule, because it uses the $K$ operator distributivity over implication:
\begin{lemma}
Jurisdiction rule is a correct rule in \textit{DELP} system.
\begin{align*}
     \frac{i \believes j \controls m \ \ i \believes j \believes m}{i \believes m}
\end{align*}
\end{lemma}
\begin{proof} 
We will prove this using Lean. 
\begin{lstlisting}[language=Lean, numbers=none, basicstyle=\small]
lemma JR_is_correct (σ : ℕ) { m i j : message σ } { Γ : ctx σ }
  : (σ-Γ ⊢ (K i, (K j, ι m) -> (ι m)) ∧ (K i, K j, ι m)) -> (σ-Γ ⊢ K i, ι m) :=
  λ h, mp 
    (mp kdist $ andleft h)
    (andright h). 
\end{lstlisting}
\end{proof}

\section{\textit{Needham-Schroeder} protocol implementation in Lean}\label{nsprotocol}

In this section we will analyze the \textit{Needham-Schroeder} protocol and we will implement the specification in \textit{Lean}, in order to prove some security properties. 
We recall the exchange of messages in \textit{Needham-Schroeder} protocol:
\begin{align*}
    A \to S: &\ A, B, N_a \\
    S \to A: &\ \{ N_a, B, K_{ab}, \{ K_{ab}, A \}_{K_{bs}} \}_{K_{as}} \\
    A \to B: &\ \{ K_{ab}, A \}_{K_{bs}} \\ 
    B \to A: &\ \{ N_b \}_{K_{ab}} \\ 
    A \to B: &\ \{ N_b - 1 \}_{K_{ab}}
\end{align*}

\subsection{Protocol description in Lean}
In this subsection we will formalize the specification in \textit{DELP} and then we will implement every \textit{DELP} formula in Lean.

\ 

\textbf{First step: intialization}

\

The initial knowledge of agents are:
\begin{align}
    &K_A (@N_A \land @key_{K_{AS}}(A, S)) \\ 
    &K_S (@key_{K_{AS}}(A, S) \land @key_{K_{BS}}(B, S) \land @key_{K_{AB}}(A, B)) \\ 
    &K_B @key_{K_{BS}}(B, S)
\end{align}
In Lean we have:
\begin{lstlisting}[language=Lean, numbers=none, basicstyle=\small]
axiom NSinit (σ : ℕ) { Γ : ctx σ } { A B S Na Kab Kas Kbs : message σ } 
  : σ-Γ ⊢ (K A, ((ι Na) ∧ (ι Kas.keys A S))) 
    ∧ (K S, ((ι Kas.keys A S) ∧ (ι Kbs.keys B S) ∧ (ι Kab.keys A B))) 
    ∧ (K B, (ι Kbs.keys B S)).
\end{lstlisting}

\ 

\textbf{First round: exchange of messages between A and S}

\ 

In \textit{DELP} we have:
\begin{align}
    [send_A][recv_S]@N_A
\end{align}
with the corresponding Lean implementation:
\begin{lstlisting}[language=Lean, numbers=none, basicstyle=\small]
axiom NS₁AtoS (σ : ℕ) { Γ : ctx σ } { A S Na : message σ } 
  : σ-Γ ⊢ [send A][recv S](ι Na). 
\end{lstlisting}

\ 

\textbf{Second round: exchange of messages between S and A}

\begin{align}
    &[send_S][recv_A]\bigg{(}@\{N_A\}_{K_{AS}} \land @\{key_{K_{AB}}(A, B)\}_{K_{AS}} \\ & \nonumber \ \ \ \ \ \ \ \ \ \ \land @\{\{key_{K_{AB}}(A, B)\}_{K_{BS}}\}_{K_{AS}}\bigg{)}
\end{align}
\begin{lstlisting}[language=Lean, numbers=none, basicstyle=\small]
axiom NS₂StoA (σ : ℕ) { Γ : ctx σ } { A B S Na Kab Kas Kbs : message σ }
  : σ-Γ ⊢ [send S][recv A]((ι {Na}Kas) 
    ∧ (ι {(Kab.keys A B)}Kas) 
    ∧ (ι {{(Kab.keys A B)}Kbs}Kas)).
\end{lstlisting}

\ 

\textbf{Third round: exchange of messages between A and B}

\ 

This is the last round we can formalize using \textit{DELP} system at the moment. For the next two round, we need a more expressive system, that can model both the knowledge and belief. However, up to this point we can prove that $K_{ab}$ is a common secret between $A$ and $B$, but we cannot prove the mutual authentication of these two agents.
\begin{align}
    [send_A][recv_B]@\{key_{K_{AB}}(A, B)\}_{K_{BS}}
\end{align}
\begin{lstlisting}[language=Lean, numbers=none, basicstyle=\small]
axiom NS₃AtoB (σ : ℕ) { Γ : ctx σ } { A B S Kab Kbs : message σ }
  : σ-Γ ⊢ [send A][recv B]ι {(Kab.keys A B)}Kbs. 
\end{lstlisting}

\subsection{Verifying security properties of Needham-Schroeder}
In order to prove some security properties, we must prove the following lemma that we will use further.
\begin{lemma}
Let $\Gamma$ be a set of statements, $i$ and $j$ two agents and $\varphi$ a formula. Then $\Gamma \vdash [send_i][recv_j]\varphi$ implies $\Gamma \vdash K_j \varphi$.
\end{lemma}
\begin{proof}
We will prove this lemma using Lean.
\begin{lstlisting}[language=Lean, numbers=none, basicstyle=\small]
lemma secv_imp_knowledge (σ : ℕ) { Γ : ctx σ } { i j : message σ } { φ : form σ }
  : (σ-Γ ⊢ [send i][recv j]φ) -> (σ-Γ ⊢ K j, φ) := 
  λ h, mp knowreceivef
    $ mp ktruth 
      $ mp knowsendf h. 
\end{lstlisting}
\end{proof}

We can prove that the agent $A$ knows the communication key between $A$ and $B$.
\begin{theorem}
In \textit{Needham-Schroeder} protocol, the agent A knows the communication key between A and B.
\end{theorem}
\begin{proof}
We will prove this theorem using Lean.
\begin{lstlisting}[language=Lean, numbers=none, basicstyle=\small]
theorem A_knows_Kab (σ : ℕ) { Γ : ctx σ } { A B S Na Kab Kas Kbs : message σ }
  : σ-Γ ⊢ K A, ι(Kab.keys A B) :=
  kgen 
    $ mp pdtruth 
      $ mp honestyright 
        $ andintro 
          (mp ktruth $ A_knows_Kas A B S Na Kab Kas Kbs)
          (mp ktruth $ A_knows_Kab_encrypted_Kas A B S Na Kab Kas Kbs).
\end{lstlisting}
\end{proof}

In a similar way, we can prove that also $B$ knows the communication key between $A$ and $B$.
\begin{theorem}
In \textit{Needham-Schroeder} protocols, the agent B knows the communication key between A and B.
\end{theorem}
\begin{proof}
We will prove this theorem using Lean.
\begin{lstlisting}[language=Lean, numbers=none, basicstyle=\small]
theorem B_knows_Kab { σ : ℕ } { Γ : ctx σ } { A B S Na Kab Kas Kbs : message σ }
  : σ-Γ ⊢ K B, ι(Kab.keys A B) :=
kgen $ mp pdtruth
  $ mp honestyright
    $ andintro 
      (mp ktruth $ B_knows_Kbs A B S Na Kab Kas Kbs) 
      (mp ktruth $ secv_imp_knowledge $ NS₃AtoB A B S Kab Kbs).
\end{lstlisting}
\end{proof}

We have now that $K_{ab}$ is a common secret between $A$ and $B$, but we cannot prove that we also have a mutual authentication. We know that $K_A @key_{K_{ab}}(A, B) \land K_B @key_{K_{ab}}(A, B)$, but we don't know if $K_A K_B @key_{K_{ab}}(A, B)$ and $K_B K_A @key_{K_{ab}}(A, B)$.

\section{Conclusion and further work}


The system \textit{DELP} is closely related to the system \textit{POL} (Public observation logic \cite{halpern2017epistemic}), but it has a different semantics for $[\alpha]\varphi$: the  updated models of \textit{POL} are replaced by \textit{DEL} models  \cite{van2007dynamic}, while the set  $Exp$ represents  the "adversary knowledge" (defined as in the operational semantics from \cite{hollestelle2005systematic})  and not the  "expected observations" (as in \textit{POL}).  Even if our system is simpler than the one from \cite{halpern2017epistemic}, we are able to  translate \textit{BAN} logic and to validate \textit{BAN} inference rules. 

Our work so far shows that \textit{DELP} is a good candidate for modelling and analysing security protocols. We are aimig to define a system that has a rigourous theoretical development: it is complete and  all proofs are certified  by  Lean implementations.  

At this stage we've already noticed that further refinements are needed:  so far we used "knowledge" operators but, in order to increse our system expressiveness, we would like to  model the epistemic  "trust"; we also consider adding a temporal behaviour, in order to be able to model the property of  \textit{freshness} since, currently, we use a weaker variant, namely the uniqueness on the system (\textit{nonce}). Last but not least, we consider adding the probabilistic interpretation, following the initial idea from  \cite{halpern2017epistemic}.


On the implementation side in \textit{Lean}, we will add the proof for the completeness theorem and we will keep all the theoretical results  automatically verified for any subsequent modification.

\end{document}